\documentclass[12pt]{amsart}
\usepackage{bbm}
\usepackage{mathrsfs}
\usepackage{amsfonts}
\usepackage{amscd}
\usepackage{amsmath}
\usepackage{amsthm}
\usepackage{amsfonts}
\usepackage{amssymb}
\usepackage{graphicx}
\usepackage{amsmath,amsthm,hyperref} 
\usepackage{fancyhdr}
\numberwithin{equation}{section}

\newtheorem{theorem}{Theorem}[section]

\newtheorem{proposition}[theorem]{Proposition}

\def\D(a,r){\Delta(a,r)}
\def\D\m\D(a,r_0){\Delta\setminus\Delta(a,r_0)}
\def\D{\Delta}

\def\D{\Bbb D}

\newcommand{\dd}[1]{\, {\mathrm d}{#1}}
\newcommand{\sinc}{\mathrm{sinc}}
\newcommand{\e}{\mathrm{e}}
\newcommand{\im}{\mathrm{i}}

\begin{document}
\title{Sampling Error Analysis and  Some Properties of
   Non-bandlimited Signals  That Are Reconstructed by Generalized Sinc Functions}
\author{Youfa Li, Qiuhui Chen,  Tao Qian and Yi Wang}

\address{Youfa Li, College of
Mathematics and Information Sciences, Guangxi University, Nanning,
CHINA; Department of Mathematics, Faculty of Science and Technology,
University of Macau, Taipa, Macao, CHINA.}
\email{youfalee@hotmail.com}
\address{Qiuhui Chen, Cisco School of Informatics, Guangdong University of
Foreign Studies, Guangzhou,  CHINA.} \email{chenqiuhui@hotmail.com}

\address{Tao Qian, Department of Mathematics, Faculty of Science
and Technology, University of Macau, Taipa, Macao, CHINA.}
\email{fsttq@umac.mo}

 \address{Yi Wang, corresponding author, Department of Mathematics, Auburn University at Montgomery,
 P.O. Box 244023, Montgomery, AL 36124-4023
 USA.  } \email{ywang2@aum.edu}

\begin{abstract}
Recently efforts have been made to use generalized sinc functions to perfectly reconstruct various kinds of non-bandlimited signals. As a consequence, perfect reconstruction sampling formulas have been established using such generalized sinc functions.
This  article  studies the error of the reconstructed non-bandlimited signal when an adaptive truncation scheme is employed. Further, when there are noises present in the samples, estimation on  the expectation and variance of the error pertinent to the reconstructed signal is also given.    Finally discussed are   the reproducing properties and the Sobolev smoothness of functions in the space of non-bandlimited signals that admits  such a sampling formula.

\end{abstract}

 \subjclass[2000]{}
\thanks{Li is supported by the
National Natural Science Foundation of China (Grant No.11126343),
Natural Scientific Project of Guangxi University under grant
XBZ110572 and by Macao Science and Technology Fund FDCT/056/2010/A3
for his Postdoctoral research. Chen is supported by NSFC under grant 61072126 and by
Natural Science Foundation of Guangdong Province under grant
S2011010004986. Qian is supported by Grant of University of Macau
UL017/08-Y3/MAT/QT01/FST and by Macao Science and Technology Fund
FDCT/056/2010/A3.}

\maketitle

\bigskip
\section{Introduction}
\medskip
We begin by establishing some  notations used in the paper.
Let $\mathbb{N}$ be the set of natural numbers, $\mathbb{Z}$ be the set of integers, and $\mathbb{Z}_+:=\{0\}\cup \mathbb{N}$. For a positive integer $n\in \mathbb{N}$, we   use the index set $\mathbb{Z}_n:=\{0, 1, \ldots, n-1\}$. Furthermore, we denote by $\mathbb{R}$  the set of real numbers, and by $\mathbb{C}$ the set of complex numbers. Let $X$ be a subset of $\mathbb{R}$, and for $q\in \mathbb{N}$,  we say a function $f$ is in $L^q(X)$ if and only if
$$
\|f\|_{q,X}:=\left(\int_X |f(t)|^q \dd{t}\right)^{1/q}<\infty,
$$
and  $f$ is said to be in $L^\infty(X)$ if
$$
\|f\|_{\infty,X}:={\mathrm{ess}}\sup \{|f(t)|: t\in X\}<\infty.
$$
Similarly, let $Z$ be a subset of $\mathbb{Z}$, a sequence ${\boldsymbol y}:=(y_k:k\in Z) $ is said to be in $l^q(Z)$ if and only if
$$
\|{\boldsymbol y}\|_{q,Z}:=\left(\sum_{k\in Z} |y_k|^q\right)^{1/q}<\infty,
$$
and ${\boldsymbol y}  $ is in $l^\infty(Z)$ if
$$
\|{\boldsymbol y}\|_{\infty,Z}:=\sup\{|y_k|:k\in Z\}<\infty.
$$

In digital signal processing,
the classic  {\em sinc} function
is fundamentally significant due to
the Whittaker-Kotelnikov-Shannon
(WKS) sampling theorem
\cite{Sha1948,Sha1949,BER1986}. Recall that the classic sinc function is defined at   $t\in \mathbb{R}$    by the equation
 $$
\mathrm{sinc} (t):=\frac{\sin t}{t}.$$

The WKS sampling theorem enables to reconstruct a bandlimited signal from shifts of sinc functions weighted by the uniformly spaced samples of that signal.   It is natural to ask whether similar sampling theorem exists for {\em non-bandlimited signals}. To that end, recently efforts have been made to extend the classic sinc to  {\em generalized sinc functions}, for example, in \cite{CMW2010,ChenQian-AA09,Chen-Wang-Wang}. One kind of  {generalized sinc functions} given in \cite{CMW2010}, denoted by
$\sinc_F$, is defined as the {\em inverse Fourier transform} of a so-called {\em symmetric cascade filter}, denoted by  $H $. The symmetric cascade filter $H$ is a  piecewise constant function whose value at $\xi\in \mathbb{R}$ is given by
\begin{equation}\label{eqn:general-H}
H (\xi):=\sum_{n\in \mathbb{Z}_+}b_n \chi_{I_n}(\xi),
\end{equation}
where  the sequence  $\boldsymbol
b=\left( b_n:n \in \mathbb{Z}_+\right)$ is   in
$l^2({\mathbb Z}_+)$, $\chi_I$ is the {\em indicator function} of the set $I$, and the interval $I_n$, $n\in \mathbb{Z}_+$, is the union of two symmetric intervals  given by the equation
$$
{I}_{n} := (-(n+1) ,-n ]\cup[n ,(n+1) ).
$$
Thus the generalized sinc function $\sinc_F$ is defined  by the equation 
\begin{equation}\label{general-sinc-function}
\sinc_F:=\sqrt{\frac{\pi}{2}}{\mathcal F}^{-1}H ,
\end{equation}
where for any signal $f \in L^2(\mathbb{R}) $ and $\xi\in{\mathbb R}$
$$
({\mathcal F}f)(\xi)=\hat{f}(\xi):=\frac{1}{\sqrt{2\pi}}\int_{\mathbb R}
f(t)\e^{-\im\xi t}\dd{t}.
$$
Of course, we have that $H\in L^2(\mathbb{R})$ because $\boldsymbol
b \in l^2({\mathbb Z}_+) $, and hence $\sinc_F\in L^2(\mathbb{R})$ since the Fourier operator is closed in  $L^2(\mathbb{R})$.

With the generalized function $\sinc_F$,  a perfect reconstruction sampling theorem was established  in \cite{CMW2010}  for the purpose of reconstructing {non-bandlimited signals}.  This kind of reconstruction  sampling theorem may be very useful to study signals with polynomial decaying Fourier spectra that arise in evolution equations and control theories \cite{LRa2005, Phu2007}.

One of our goals in this paper is to study the error of the reconstructed non-bandlimited signal when an adaptive truncation scheme is employed. This will be done in Section \ref{sec:samplingerror}. We further analyze the expectation and variance of the error of the reconstructed signal    when there are noises present in the samples in Section \ref{sec:noiseana}.  Finally we discuss  the reproducing properties and Sobolev smoothness of functions in the space of non-bandlimited signals that admits such a sampling formula. We begin with a discussion of the construction and known properties of the function $\sinc_F$ in Section \ref{sec:review}.

\section{Properties of the generalized sinc functions}\label{sec:review}
Surprisingly the function $\sinc_F$ has many properties that are similar to the classic sinc, such as cardinal, orthogonal properties and it behaves also similarly to the classic sinc. In the special case, the function $\sinc_F$   reduces to the classic $\sinc$. Let us first review the approach to obtain an explicit form of $\sinc_F$.

The symmetric cascade filter $H $ can be associated with an analytic function $F$  on the open unit disk $$\Delta:=\{\zeta\in \mathbb{C}: |\zeta|<1  \}.$$  The  value of $F$ at   $z \in \Delta$   is well defined by
\begin{equation}\label{analfu}
F(z):= \sum_{n\in\mathbb{Z_+}}b_nz^n,
\end{equation}
as $\boldsymbol
b \in l^2({\mathbb Z}_+) $.
Recall that the Hardy space $H^2(\Delta)$ consists of all functions $f$ analytic in $\Delta$, with norm given by
$$
\|f\|^2_{H^2(\Delta)}=\sup_{r\in (0,1)} \frac{1}{2\pi} \int_{[-\pi, \pi]}|f(r\e^{\im t})|^2 \dd{t}.
$$
Since, by hypothesis,  $\boldsymbol b\in l^2({\mathbb Z}_+)$, we have that $F\in H^2(\Delta)$. Consequently its extension to the boundary $\partial \Delta$   of $\Delta$  is in $L^2(\partial\Delta)$.

Thus, from equations \eqref{general-sinc-function}, \eqref{eqn:general-H} and \eqref{analfu} an explicit form of
$\sinc_F(t)$,  $t\in\mathbb R$ can be found as
\begin{equation}\label{phi-explicit-form}
\sinc_F(t)=\mathrm{sinc}\left(\frac{  t}{ 2}\right)
\mathrm{Re}\left\{F(\e^{\im t})\e^{\frac{1}{2}\im  t} \right\} ,
\quad \mathrm{a.e.}
\end{equation}
where $\mathrm{Re}(z)$ is the real part of a complex number $z$.

We observe that if $\boldsymbol{b}\in l^1(\mathbb{Z}_+)$ then $H\in L^1(\mathbb{R})$ and $F$ is continuous on the boundary of $\Delta$, which in turn implies $\sinc_F$ is continuous and   bounded.

A very interesting fact, as discovered in the paper \cite{CMi2011}, is that when  $F$ is imposed with a {\em stronger condition} of having analyticity in a neighborhood of the closed unit disc $\overline{\Delta}$, the function $\sinc_F$ can be generated through a function, denoted by $G$, that is also analytic in a neighborhood of the closed unit disc,  real on the real axis and normalized so that $G(1)=1$ and $G'(1)\ne 0$.
The function $G$ is linked to $F$ by the equation
\begin{equation}\label{eqn:def_F}
F(z):=\frac{G(z)-1}{z-1}.
\end{equation}
A real-valued function $\phi_G $ whose value at $t\in \mathbb{R}$ is then defined through the imaginary part of the values of  $G$  on the unit circle by the equation
\begin{equation}\label{eqn:eqn:v_G}
 \phi_G(t):=\frac{\mathrm{Im}(G(\e^{\im t}))}{t}.
 \end{equation}

Applying equation \eqref{phi-explicit-form} to compute $\sinc_F$ by using   equations \eqref{eqn:def_F} and \eqref{eqn:eqn:v_G}   yields for all $t\in \mathbb{R}$,
\begin{equation} \label{eqn:phiF}
\sinc_F(t)=\phi_G(t) .
\end{equation}

Two important examples can be demonstrated for this construction. When
  $G=z$, i.e., $F=1$, we have   $\mathrm{sinc}_F=\phi_G=\mathrm{sinc}$.
For the second example, let $G$ be the {\em Blaschke product} of order $n\in \mathbb{N}$ with parameters $\boldsymbol{a}:=(a_j: j\in \mathbb{Z}_n) \in (-1,1)^n$, that is,
$$
G(z)=B_{\boldsymbol{a}}(z):=\prod_{j\in \mathbb{Z}_n } \frac{z-a_j}{1-a_j z}.
$$
Then $\mathrm{sinc}_F(t)=\phi_G(t)=  \frac{\sin \theta_{\boldsymbol{a}} (t)}{t}$, where $\theta_{\boldsymbol{a}}$ is determined by the boundary value of the Blaschke product at $t\in \mathbb{R}$ by
$
\e^ {\im \theta_{\boldsymbol{a}}(t)}=B_{\boldsymbol{a}}(\e ^{\im t}).
$

We next list some properties of the function $\sinc_F$.
\begin{proposition}\label{prop:sinc}
Let  the generalized function $\sinc_F$ be defined by equation \eqref{phi-explicit-form}.   Then
\begin{enumerate}

\item   $\sinc_F(n\pi)=F(1)\delta_{n,0}$, where $\delta_{n,0}=1$ if $n=0$ and $\delta_{n,0}=0$ if $n\in \mathbb{Z}\setminus \{0\}$.
    \item $\sinc_F$ is bounded, infinitely differentiable.

  \item $|\sinc_F(t)|\le \frac{4\|{\boldsymbol b}\|_{l^1(\mathbb{Z}_+)}}{2+|t|}$, for $t\in \mathbb{R}$, and $ \sinc_F \in L^2(\mathbb{R})$.

    \item The set $\{ \sinc_F(\cdot -n\pi): n\in \mathbb{Z}  \}$ is an orthogonal set, that is
        $$
        \langle \sinc_F, \sinc_F(\cdot -n\pi)  \rangle =\pi\|{\boldsymbol b}\|^2_{l^2(\mathbb{Z_+})} \delta_{n,0},
        $$
        where $\langle \cdot, \cdot \rangle $ is the usual inner product on the Hilbert space $L^2(\mathbb{R})$.
\end{enumerate}
\end{proposition}

\begin{proof}
The first two statements directly follow from equation \eqref{phi-explicit-form}.
 The third statement follows from equation  \eqref{phi-explicit-form} and noticing $  \sinc (t) \le \frac{2}{1+|t|}$ for any $t\in \mathbb{R}$.
The fourth statement is a special case of Corollary 3.2 of \cite{CMW2010}. For the convenience of readers, we provide a direct proof here. By Parseval's theorem and    equation  \eqref{general-sinc-function}  we have
\begin{eqnarray*}
\int_\mathbb{R} \sinc_F(t)\sinc_F(t-n \pi)\dd{t} &=&  \frac{\pi}{2}\int_\mathbb{R} H^2(x) \e ^{\im n\pi x}\dd{x}
= \frac{\pi}{2}\int_\mathbb{R} \sum_{k\in \mathbb{Z}_+}b_k^2 \chi_{I_k}(x) \e ^{\im n\pi x}\dd{x}\\
&=& \frac{\pi}{2}\sum_{k\in \mathbb{Z}_+}b_k^2  \int_{I_k} \e ^{\im n\pi x}\dd{x}
= \pi \left(\sum_{k\in \mathbb{Z}_+}b_k^2\right) \delta_{n,0},
\end{eqnarray*}
where, in the last equality we have used the orthogonality of the set $\{\e^{-\im n\pi \xi} : n\in \mathbb{Z}  \}$ on $I_k$, $k\in \mathbb{Z}_+$.
The interchange of the integral operator and the infinite sum is guaranteed by the convergence of the series.

\end{proof}

\section{Sampling Truncation Error Analysis}\label{sec:samplingerror}
In \cite{CMW2010}, a Shannon-type sampling theorem  is given concerning functions in the shift-invariant space   $$V_F:=\left\{\sum_{n\in \mathbb{Z}} c_n \sinc_F (\cdot -n\pi): F(1)=1,\boldsymbol{c}=(c_n: n\in \mathbb{Z})\in l^2 (\mathbb{Z})  \right\}$$
of $\sinc_F$.   The Shannon-type sampling theorem  is the direct result of the properties in the previous proposition. We record it here.
\begin{theorem} \label{thm:sampling}
A signal $f\in V_F$ if and only if
\begin{equation}\label{eqn:sampling}
f=\sum_{n\in \mathbb{Z}} f(n\pi) \sinc_F (\cdot -n\pi).
\end{equation}
\end{theorem}
Equation \eqref{eqn:sampling} necessarily implies that the sampling sequence $(f(n\pi):n\in \mathbb{Z})\in l^2 (\mathbb{Z})$ by the orthogonality of the set $\{\sinc_F(\cdot-n\pi):n\in \mathbb{Z}  \} $. The above equation of course is true in $L^2(\mathbb{R})$ norm. However,    if ${\boldsymbol b}\in l^1(\mathbb{Z}_+)$, equation \eqref{eqn:sampling} holds true pointwise, because by Cauchy-Schwartz inequality the series on the right side of equation \eqref{eqn:sampling} converges uniformly,  hence the limiting function $f$ is continuous.

We remark that, as pointed out in \cite{CMW2010}, a function  $f\in V_F$ can be characterized by its spectrum. Specifically, a function $f\in V_F$   if and only if
\begin{equation}\label{eqn:spectra}
\mathcal{F}f(\xi)=\sqrt {\frac{\pi}{2}} \left(\sum_{n\in \mathbb{Z}} f(n\pi) \e^{-\im  n\pi \xi} \right)H(\xi).
\end{equation}
Equation \eqref{eqn:spectra} holds true in $L^2(\mathbb{R})$ if ${\boldsymbol b}\in l^2(\mathbb{R})$, and a.e. pointwise if $f\in L^1(\mathbb{R})$ and the sample sequence $(f(n\pi):n\in \mathbb{Z})\in l^1(\mathbb{Z})$.

The following property is true for functions in the space $V_F$ that is similar to    Parseval's identity.
\begin{proposition} If $f\in V_F$ then
\begin{equation}\label{eqn:parseval}
\| f\|^{2}_{L^{2}(\mathbb{R})}=\pi\|{\boldsymbol b}\|^2_{l^2(\mathbb{Z}_+)}\sum_{n\in \mathbb{Z}} f^2(n\pi)
\end{equation}
\end{proposition}
\begin{proof}
This is a direct result of equation \eqref{eqn:sampling} and Proposition \ref{prop:sinc} (4).
\end{proof}

The recovering formula in \eqref{eqn:sampling} involves  an infinite
sum. In practice, we need to truncate the series to
approximate $f$.
 Here, we prefer an {\em adaptively truncated sum} and offer a pointwise estimation of the error.

For fixed
$n\in \mathbb{N}$ and $t\in \mathbb{R}$, define the index set $$J_n(t):=\{j: j\in \mathbb{Z}, |t-j\pi|\le n\pi\}$$ and the  partial sum  given at $t\in \mathbb{R}$ by
\begin{align}\label{tongxue}
S_{n}(t):=\displaystyle\sum _{j\in J_n(t) }
  f(j \pi ) {\sinc}_F(
t-j\pi).
\end{align}
The adaptive truncation strategy allows that for any $t\in \mathbb{R}$ and $n\in \mathbb{N}$, there are approximately $2n$    functions $\sinc_F$ shifted by a distance of $j\pi$, $j\in J_n(t)$, on both sides of  $t$.  The following theorem
states that the truncated error estimate   is
$\mathcal {O}(n^{-1/2})$. We recall a Calculus fact that will be used a couple of times later. For a positive and decreasing sequence ${\boldsymbol a}:=(a_n:n\in \mathbb{Z}_+)$, if  $\frac{1}{a_{n+1}}-\frac{1}{a_{n}}=c$, where $c$ is a positive constant, then
\begin{equation}\label{eqn:seqest}
\sum_{n\in \mathbb{Z}_+} \frac{1}{a_n^2}\le \frac{1}{a_0^2}+\frac{1}{c}\int_{[a_0,\infty)}\frac{1}{x^2}\dd{x}.
\end{equation}

\begin{theorem}
Let   $f\in V_F$ and ${\boldsymbol b}\in l^1(\mathbb{Z}_+)$. Then we have
\begin{align}\label{eqn:superror}
\sup_{t\in \mathbb{R}}|f(t)-S_{n}(t)|\leq \frac{\|\boldsymbol{b}\|_{\ell^{1}(\mathbb{Z}_+)} }{\|\boldsymbol{b}\|_{\ell^{2}(\mathbb{Z}_+) }}
  \ \Vert
f\Vert_{L^{2}(\mathbb{R})}\sqrt{\frac{8}{\pi^3}\left(\frac{1}{n^2}+\frac{1}{n}\right)
  }.
\end{align}
\end{theorem}

\begin{proof}
We first note that by equation \eqref{phi-explicit-form},   for   $t\in \mathbb{R}$,
\begin{equation}\label{eqn:sincest}
|\sinc_F(t)|\le \|\boldsymbol{b}\|_{l^1(\mathbb{Z}_+)}\left|\sinc\frac{t}{2}\right|\le \left\{ \begin{array}{l l}
\|\boldsymbol{b}\|_{l^1(\mathbb{Z}_+)}  & \mbox{ if } |t|\le 2,\\\\
\frac{\|\boldsymbol{b}\|_{l^1(\mathbb{Z}_+)} }{|t/2|} &\mbox{ if } |t|\ge 2.
\end{array}\right.
\end{equation}
By    Cauchy-Schwartz inequality and  equation \eqref{eqn:sincest},
we have
\begin{eqnarray}
|f(t)-S_{n}(t)|^2
&=&\left|\displaystyle\sum_{j\in \mathbb{Z}\setminus J_n(t)} f(j {\pi}{ }
)\sinc_F(  t-j\pi)\right|^2
\notag\\
&\leq& \left(\displaystyle\sum_{j\in \mathbb{Z}\setminus J_n(t)}f^{2}
(j {\pi}{ })\right) \left(\displaystyle\sum_{j\in \mathbb{Z}\setminus J_n(t)}
\sinc ^2_F(  t-j\pi)\right)
\notag\\
&\leq& 4{ \|\boldsymbol{b}\|^2_{\ell^{1}(\mathbb{Z}_+)}} \left(\displaystyle\sum_{j\in \mathbb{Z}\setminus J_n(t)}f^{2}
(j {\pi}{ })\right) \left(\displaystyle\sum_{j\in \mathbb{Z}\setminus J_n(t)}
\frac{1}{  (t- {j\pi}{ })^{2}}\right) .\label{dvf}
\end{eqnarray}
We next estimate for $t\in \mathbb{R}$ the value of
\begin{align}\label{1dvf} {R}_{n}(t):=\displaystyle\sum_{j\in \mathbb{Z}\setminus J_n(t)}
 \frac{1}{ (t- {j\pi}{ })^2} . \end{align}
It is easy  to see that $R_n$ is periodic with period $\pi$. For $t\in [0,\pi)$, using equation \eqref{eqn:seqest} we have
\begin{eqnarray}\label{eqn:R_n}
R_n(t)& \le & 2 \left( \frac{1}{n^2\pi^2}+\frac{1}{\pi}\int_{[n\pi,\infty)}\frac{1}{x^2}\dd{x} \right) =\frac{2}{\pi^2}\left(\frac{1}{n^2}+\frac{1}{n}\right).
\end{eqnarray}
Using equation \eqref{eqn:parseval},  we obtain that  the quantity
\begin{eqnarray}\label{eqn:partialsum}
  \sum_{j\in \mathbb{Z}\setminus J_n(t)}f^{2}
(j {\pi}{ })    \le \sum_{n\in \mathbb{Z}} f^2(n\pi) = \frac{\|f\|^2_{L^2(\mathbb{R})}}{\pi \|\boldsymbol{b} \|^2_{l^2(\mathbb{Z}_+)}}.
\end{eqnarray}
Finally we conclude \eqref {eqn:superror} by  combining equations \eqref{dvf}, \eqref{eqn:R_n} and \eqref{eqn:partialsum}.
\end{proof}

\section{Error analysis when noises present}\label{sec:noiseana}
Theorem \ref{thm:sampling}    establishes that a signal in the space $V_F$ can be perfectly reconstructed  by an infinite sum of shifts of the generalized sinc functions weighted by equally spaced samples of that signal. However, samples are often corrupted by
noise in practice.  In \cite{SZh2004}, Smale and Zhou gave an error
estimate in the probability sense  for  Shannon sampling theorem
with    noised samples. In \cite{ALS2008}, Aldroubi, Leonetti and
Sun studied  the error by  frame reconstruction from noised samples. In
this section, we shall investigate the error of the sampling formula
\eqref{eqn:sampling} with noised  samples. Specifically, we  deal
with the following noise model concerning the noisy samples $\widetilde{f}(n {\pi}{ })$, $n\in \mathbb{Z}$ whose value  corrupted by noise is given by
\begin{align}\label{fgvc}\widetilde{f}(n {\pi}{ })=f(n {\pi}{ })+\epsilon(n {\pi}{ }), \quad n\in\mathbb{Z},\end{align}
where we assume that    $(\epsilon(n\pi): n\in \mathbb{Z})$ is   a sequence of   {\em independent and identically distributed }   random variables with   the expectation and
variance  of each  given by
\begin{align}\mathrm{E}(\epsilon(n\pi))=0 ,  \ \mbox{Var}(\epsilon(n\pi))=\sigma^{2},\qquad n\in \mathbb{Z}.\end{align}
Thus in practice, we recover $f\in V_F$ by \begin{align}\label{drvi}f^{\natural}=\displaystyle\sum_{n\in
\mathbb{Z}} \widetilde{f}(n {\pi}{ } )
\sinc_F (  \cdot-n\pi).\end{align}  We next   study the
expectation $\mathrm{E}(f-f^{\natural})$ and variance
$\mbox{Var}(f-f^{\natural})$.

\begin{theorem}\label{thm:noiseerr}
Let $f\in V_F$ be recovered by
equation \eqref{drvi} with the noised  samples
$\widetilde{f}(n {\pi}{ })$ being referred to in \eqref{fgvc}.
Then
$$\mathrm{E}(f(t)-f^{\natural}(t))=0$$ and
\begin{equation}\label{eqn:var}
\mathrm{Var}(f(t)-f^{\natural}(t))\leq 2\sigma^2\|\boldsymbol{b}\|^2_{l^1(\mathbb{Z}_+)} \left(1+ \frac{8}{\pi^2}\right).\end{equation}
\end{theorem}
\begin{proof}
We first compute the expectation $\mathrm{E}(f(t)-f^{\natural}(t))$.
\begin{eqnarray}
\mathrm{E}(f(t)-f^{\natural}(t))&=&\mathrm{E}\left(\displaystyle\sum_{n\in
\mathbb{Z}} \epsilon(n {\pi}{ }
)\sinc_F(  t-n\pi)\right) \notag\\
&=&\displaystyle\sum_{n\in
\mathbb{Z}} \mathrm{E}(\epsilon(n {\pi}{ }
))\sinc_F(  t-n\pi) \notag\\
&=&0. \label{aya8}\notag
\end{eqnarray}
Invoking the assumed independence of
  $\epsilon(j {\pi}{ } )$,
$  j\in \mathbb{Z} $ we obtain
\begin{eqnarray}
\mbox{Var}(f(t)-f^{\natural}(t))&=\mathrm{Var}\left(\displaystyle\sum_{n\in
\mathbb{Z}} \epsilon(n {\pi}{ }
)\sinc_F(  t-n\pi)\right) \notag\\
&= \displaystyle\sum_{n\in
\mathbb{Z}} \mathrm{Var} \big(\epsilon(n {\pi}{ }
)\sinc_F(  t-n\pi)    \big)\notag\\
&= {\sigma^{2}} \displaystyle\sum_{n\in
\mathbb{Z}} \sinc_F^2(  t-n\pi) .  \label{aya9}
\end{eqnarray}
 The sequence $\displaystyle\sum_{n\in
\mathbb{Z}} \sinc_F^2(  t-n\pi)  $ can be easily estimated. Note it is periodic with period $\pi$.
Recalling equation \eqref{eqn:sincest}, for $t\in [0,\pi)$, we have
\begin{eqnarray*}
\displaystyle\sum_{n\in
\mathbb{Z}} \sinc_F^2(  t-n\pi)  &=& \sinc_F^2(t)+\sinc_F^2(t-\pi) +\sum_{n\in -\mathbb{N}}\sinc_F^2(t-n\pi)+\sum_{n\in 1+\mathbb{N}}\sinc_F^2 (t-n\pi)\\
&\le& 2+ \sum_{n\in -\mathbb{N}}\sinc_F^2(t-n\pi)+\sum_{n\in 1+\mathbb{N}}\sinc_F^2 (t-n\pi).
\end{eqnarray*}
Noting for $t\in [0,\pi)$, $n\in \mathbb{N}$, $\frac{|t-n\pi|}{2}\ge \frac{\pi}{2}$ and in view of equation \eqref{eqn:seqest} we obtain that
$$
\sum_{n\in 1+\mathbb{N}}\sinc_F^2 (t-n\pi) \le \|\boldsymbol{b}\|^2_{l^1(\mathbb{Z}_+)} \left( \frac{4}{\pi^2}+\frac{2}{\pi}\int_{[\frac{\pi}{2},\infty)}\frac{1}{x^2}\dd{x}\right)=
\|\boldsymbol{b}\|^2_{l^1(\mathbb{Z}_+)} \left(\frac{8}{\pi^2} \right).
$$
Similarly we have
$$
\sum_{n\in -\mathbb{N}}\sinc_F^2 (t-n\pi) \le
\|\boldsymbol{b}\|^2_{l^1(\mathbb{Z}_+)} \left(\frac{8}{\pi^2} \right).
$$
Consequently we obtain that
\begin{equation}\label{eqn:sumsinc}
\displaystyle\sum_{n\in
\mathbb{Z}} \sinc_F^2(  t-n\pi) \le 2\|\boldsymbol{b}\|^2_{l^1(\mathbb{Z}_+)} \left(1+ \frac{8}{\pi^2}\right).
\end{equation}
Finally combining equations \eqref{aya9} and \eqref{eqn:sumsinc} proves equation \eqref{eqn:var}.
\end{proof}

\section{The reproducing property  and Sobolev smoothness }\label{sec:repr}
When the analytical function $F$ is chosen to be $F=1$, the space $V_F$ reduces to the space of bandlimited signals.
 The space of bandlimited signals is a reproducing kernel
Hilbert space (r.k.H.s)  \cite{Dau}.
We next show that   the space $V_F$ has a similar property. However, the reproducing kernel is   a  distribution   in the space of tempered distributions.  We define the distribution for $x,t\in \mathbb{R}$ by
$$
\Phi(x,t):=\frac{1}{\pi}\sinc_F(t-x)\sum_{k\in \mathbb{Z}}\e^{\im 2k x}.
$$
Recalling the Poisson formula in the distribution form:
$$
\sum_{k\in \mathbb{Z}} \e^{-\im 2k x}=  \pi\sum_{k\in \mathbb{Z}} \delta (x-k\pi),
$$
where $\delta$ is the usual Dirac delta function, we immediately obtain  an alternative form of $\Phi$ given by
\begin{equation}\label{eqn:kernel}
\Phi (x,t)=\sum_{k\in \mathbb{Z}}   \delta (x-k\pi) \sinc_F(t-x).
\end{equation}

\begin{theorem}
Let $f\in V_F$ then
\begin{equation}\label{eqn:recons}
f(t)= \int_\mathbb{R} f(x) \Phi(x,t)\dd{x}
\end{equation}
in the distribution sense.

\end{theorem}
\begin{proof}
Poisson's summation formula indicates that
\begin{equation}\label{eqn:Poi}
\sum_{n\in \mathbb{Z}}f(n\pi)\e^{-\im n\pi \xi}= \sqrt{\frac{2}{\pi}} \sum_{k\in \mathbb{Z}} \mathcal{F}f (\xi+2k).
\end{equation}
Thus equation \eqref{eqn:spectra} is equivalent to
\begin{equation}
\mathcal{F}f(\xi)=\sum_{k\in \mathbb{Z}} \mathcal{F}{f}(\xi+2k)H(\xi)
\end{equation}
This leads to
 \begin{eqnarray*}f(t)&=& \mathcal{F}^{-1}\left( \sum_{k\in \mathbb{Z}} \mathcal{F}{f}(\cdot+2k)H \right)(t)\\
&=& \sum_{k\in \mathbb{Z}} \mathcal{F}^{-1}\left(  \mathcal{F}{f}(\cdot+2k)H \right)(t)
\end{eqnarray*}
The interchange of the order of the sum and the integral operator is justified by the convergence of the series. Let $g*h$ be the convolution of two functions $g,h\in L^2(\mathbb{R}$). Recalling equations \eqref{general-sinc-function} and    the convolution theorem for the Fourier transform, we therefore conclude that
\begin{eqnarray*}
f(t)&=&\sum_{k\in \mathbb{Z}}\frac{1}{\sqrt{2\pi}}\Big(\left(\mathcal{F}^{-1}(\mathcal{F}{f}(\cdot+2k)\right)*
\left(\mathcal{F}^{-1}H\right)\Big)(t)\\
&=&\sum_{k\in \mathbb{Z}}\frac{1}{\pi}\int_\mathbb{R} f(x) \e^{-\im 2kx} \sinc_F(t-x)\dd{x}\\
&=& \frac{1}{\pi}\int_\mathbb{R} f(x) \sum_{k\in \mathbb{Z}}\e^{-\im 2kx} \sinc_F(t-x)\dd{x}
\end{eqnarray*}

which is    equation \eqref{eqn:recons}.

\end{proof}

Next we discuss the {\em Sobolev smoothness} of a function in the space $V_F$.
We say that a
function $f$ belongs to the {Sobolev space} $H^{s}(\mathbb{R})$ if
\begin{equation}\label{ee1}
 \displaystyle\int_{\mathbb{R}}|\hat{f}(\xi)|^{2}(1+|\xi|^{2})^{s}\dd{\xi}<\infty.
\end{equation}
The {\em Sobolev smoothness} of $f$ is defined to be
$\nu_{2}(f):=\sup\{s: f\in H^{s}(\mathbb{R})\}$. As an example, the Sobolev smoothness of a function   in the space of bandlimited signals is  infinity.   Algorithms \cite{Han2003, Han2003-2, Han2008, Li2011} or even Matlab routines \cite{Jia2001} are given for
calculating Sobolev smoothness of a refinable function. However, those
algorithms   are not applicable to calculate the
Sobolev smoothness of a function in the space $V_F$.
We give a characterization in the next theorem about the Sobolev smoothness of the functions in the space $V_F$.

\begin{theorem}
If $ f\in
V_F$ and for some $s\in \mathbb{R}$,
\begin{equation}\label{ee2}
\displaystyle\sum _{n\in \mathbb{Z_+}} b_{n}^{2}n^{2s}<\infty
\end{equation}
then the Sobolev smoothness of the function $f$ satisfies $\nu_{2}(f)\geq
s$.
\end{theorem}
\begin{proof}
By equations \eqref{eqn:general-H} and \eqref{eqn:spectra} we obtain that
\begin{eqnarray*}
  \displaystyle\int_{\mathbb{R}}|\hat{f}(\xi)|^{2}(1+|\xi|^{2})^{s}\dd{\xi}
 &=& \sum_{n\in \mathbb{Z}_+}b_n^2 \int_{I_n}\left( \sum_{k\in \mathbb{Z}}f(k\pi)\e^{-\im k\pi \xi} \right)^2(1+\xi^2)^s \dd{\xi}.
\end{eqnarray*}
Noting   when $\xi\in I_n$, $n\le |\xi|< n+1$ for each $n\in \mathbb{Z}_+$, which implies that  $1+\xi^2\le 3n^2$, for $n\in \mathbb{Z}_+$. Consequently, we deduce that
\begin{eqnarray*}
  \displaystyle\int_{\mathbb{R}}|\hat{f}(\xi)|^{2}(1+|\xi|^{2})^{s}\dd{\xi}
 &\le& \sum_{n\in \mathbb{Z}_+}3^sb_n^2n^{2s}  \int_{I_n}\left( \sum_{k\in \mathbb{Z}}f(k\pi)\e^{-\im k\pi \xi} \right)^2  \dd{\xi}\\
 &=& \sum_{n\in \mathbb{Z}_+}3^sb_n^2n^{2s} \sum_{k\in \mathbb{Z}}f^2(k\pi),
\end{eqnarray*}
where in the last  equality again we have  used the orthogonality of the set $\{\e^{-\im k\pi \xi} : k\in \mathbb{Z}  \}$ on $I_n$, $n\in \mathbb{Z}_+$.  Now apply equation \eqref{eqn:parseval} we obtain that
\begin{eqnarray*}
  \displaystyle\int_{\mathbb{R}}|\hat{f}(\xi)|^{2}(1+|\xi|^{2})^{s}\dd{\xi}
 \le & \frac{3^s}{\pi}\frac{\|f\|^2_{L^2(\mathbb{R})}}{\| \boldsymbol{b} \|^2_{l^2(\mathbb{Z}_+)}}\sum_{n\in \mathbb{Z}_+}b_n^2n^{2s}<\infty
\end{eqnarray*}
by the assumption \eqref{ee2}.
\end{proof}

For example, consider the space $V_F$ associated with the analytic function $F$ in equation \eqref{analfu} such that  the sequence ${\boldsymbol b}$ is a geometric sequence with $b_n=a^n(1-a)$, $n\in \mathbb{Z}_+$, where $
a\in (-1, 1)$.   The    series
$\displaystyle\sum^{\infty}_{n=0}n^{2s}a^{2n}(1-a)^{2}\nonumber$
 converges for every $s\in \mathbb{R}$.
Therefore,   we have the Sobolev smoothness
$\nu_{2}(f)=\infty$. In particular, when $a=0$, the space $V_F$ degenerates to the space of bandlimited functions, and its Sobolev smoothness is   infinity.

If $ {\boldsymbol b}$ is given by $ b_n= \frac{1}{(n+1)^{3}} , \ {n\in \mathbb{Z}_+}$, the series
$\displaystyle\sum^{\infty}_{n=0}n^{2s}\frac{1}{(2n+1)^{6}}$
converges if and only if $6-2s>1$, that is, $s<2.5$, which implies that
$\nu_{2}(f)=2.5$ in this case.


\end{document}